\documentclass[a4paper, titlepage, oneside, 11pt]{amsart}
\usepackage{amsmath,amssymb,amsfonts,amstext,amscd,cite,geometry,graphicx,latexsym,mathptmx,xcolor,pstricks,amscd,mathrsfs,amsthm}
\usepackage[utf8]{inputenc}
\usepackage{newunicodechar}
\usepackage[all]{xy}

\newtheorem{lem}{Lemma}[section]
\newtheorem{prop}[lem]{Proposition}
\newtheorem{thm}[lem]{Theorem}
\newtheorem{cor}[lem]{Corollary}

\def\R{{\mathbb{R}}}
\def\N{{\mathbb{N}}}
\def\Z{{\mathbb{Z}}}

\numberwithin{equation}{section}

\allowdisplaybreaks

\begin{document}

\title[Ultra-discrete Toda lattice and Pitman's transformation]{Dynamics of the ultra-discrete Toda lattice\\via Pitman's transformation}

\author[D.~A.~Croydon]{David A. Croydon}
\address{Research Institute for Mathematical Sciences, Kyoto University, Kyoto 606--8502, Japan}
\email{croydon@kurims.kyoto-u.ac.jp}

\author[M.~Sasada]{Makiko Sasada}
\address{Graduate School of Mathematical Sciences, University of Tokyo, 3-8-1, Komaba, Meguro-ku, Tokyo, 153--8914, Japan}
\email{sasada@ms.u-tokyo.ac.jp}

\author[S.~Tsujimoto]{Satoshi Tsujimoto}
\address{Department of Applied Mathematics and Physics, Graduate School of Informatics, Kyoto University, Sakyo-ku, Kyoto 606--8501, Japan}
\email{tujimoto@i.kyoto-u.ac.jp}

\begin{abstract}
By encoding configurations of the ultra-discrete Toda lattice by piecewise linear paths whose gradient alternates between $-1$ and $1$, we show that the dynamics of the system can be described in terms of a shifted version of Pitman's transformation (that is, reflection in the past maximum of the path encoding). This characterisation of the dynamics applies to finite configurations in both the non-periodic and periodic cases, and also admits an extension to infinite configurations. The latter point is important in the study of invariant measures for the ultra-discrete Toda lattice, which is pursued in the parallel work \cite{CSMontreal}. We also describe a generalisation of the result to a continuous version of the box-ball system, whose states are described by continuous functions whose gradient may take values other than $\pm1$.
\end{abstract}



\date{\today}

\maketitle

\section{Introduction}

To introduce and motivate this study of the dynamics of the ultra-discrete Toda lattice, which were first defined in \cite{Nagai} (another version of the ultra-discrete Toda lattice was described in \cite{MSTTT}), we start by recalling the related box-ball system (BBS) of \cite{takahashi1990}. We note that the link between the ultra-discrete Toda lattice and the BBS was initially presented in \cite{Nagai:sort}. In particular, states of the BBS are particle configurations $\eta=(\eta_n)_{n\in\mathbb{Z}}\in\{0,1\}^\mathbb{Z}$, where we write $\eta_n=1$ if there is a particle at site $n$, and $\eta_n=0$ otherwise. For simplicity, we begin by supposing that there is a finite number of particles, i.e.\ $\sum_{n\in\mathbb{Z}}\eta_n<\infty$. The system evolution is then described by means of a particle `carrier', which moves along $\mathbb{Z}$ from left to right (that is, from negative to positive), being initially empty until it meets the first particle, picking up a particle when it crosses one, and dropping off a particle when it is holding at least one particle and sees a space (see the left-hand side of Figure \ref{bbsfig}). Writing $\mathcal{T}\eta$ for the new configuration after one step of the dynamics, this description can be formalised in terms of the ultra-discrete Korteweg-de Vries (KdV) equation:
\begin{equation}\label{udkdv}
(\mathcal{T}\eta)_{n}=\min\left\{1-\eta_{n},\sum_{m=-\infty}^{n-1}\left(\eta_m - (\mathcal{T}\eta)_m\right)\right\},
\end{equation}
where we suppose $(\mathcal{T}\eta)_n=0$ for $n\leq \inf\{m:\:\eta_m=1\}$, so that the sums in the above definition are well-defined. An alternative characterisation of these dynamics can be given in terms of the vector $((Q_n)_{n=1}^N, (E_n)_{n=1}^{N-1})$, where $Q_n\in\mathbb{N}$ is the length of the $n$th string of consecutive particles, and $E_n\in\mathbb{N}$ is the length of string of empty spaces between the $n$th and $(n+1)$st particle string (see the right-hand side of Figure \ref{bbsfig}). It is then possible to check that after one time step of the BBS dynamics, the lengths of strings of particles and empty spaces in $\mathcal{T}\eta$ is given by a vector $(((\mathcal{T}Q)_n)_{n=1}^N, ((\mathcal{T}E)_n)_{n=1}^{N-1})\in\mathbb{N}^{2N-1}$ of the same length as  $((Q_n)_{n=1}^N, (E_n)_{n=1}^{N-1})$ that arises as the solution of the ultra-discrete Toda lattice equation:
\begin{align}
(\mathcal{T}Q)_{n} &=\min \left\{ \sum_{k=1}^n Q_k-\sum_{k=1}^{n-1}(\mathcal{T}Q)_k,E_{n}\right\},\nonumber\\
(\mathcal{T}E)_n &=Q_{n+1}+E_n-(\mathcal{T}Q)_n,\label{dynamics}
\end{align}
where for the purposes of these equations we suppose that $E_N=\infty$. Whilst the BBS naturally gives rise to vectors with integer-valued entries, the dynamics given by \eqref{dynamics} makes sense for any vector $((Q_n)_{n=1}^N, (E_n)_{n=1}^{N-1})\in(0,\infty)^{2N-1}$; it is this system that we call the ultra-discrete Toda lattice, and which will be the focus here.

\begin{figure}[t]
\centering
\includegraphics[width = 0.9\textwidth]{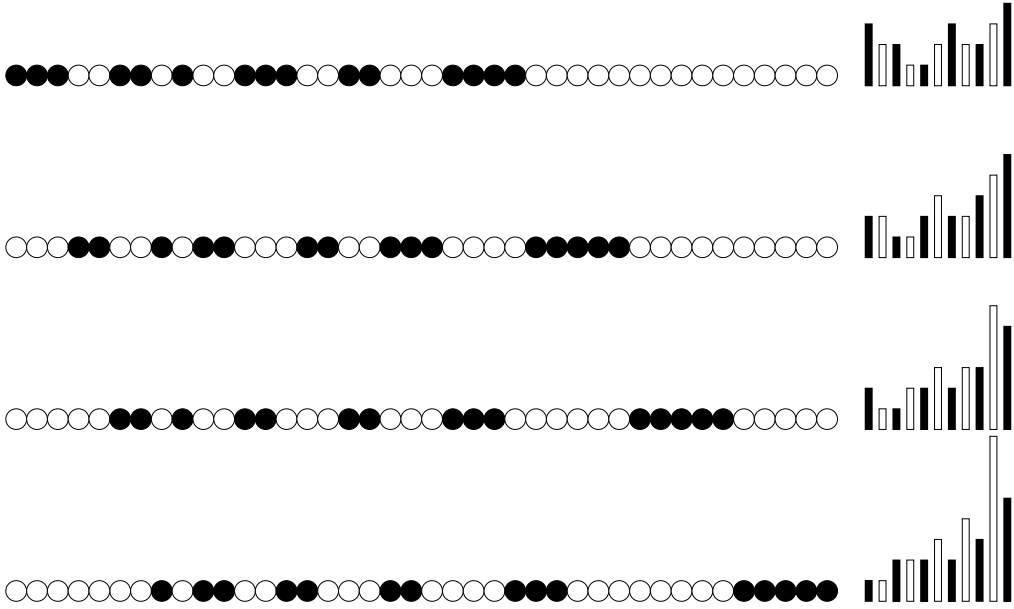}
\caption{Initial configuration and three time evolutions of the BBS (left-hand side), as well as the ultra-discrete Toda lattice representation (right-hand side).}\label{bbsfig}
\end{figure}

Towards studying the BBS from a probabilistic viewpoint, including its invariant measures, in \cite{CKST} a systematic extension of the BBS dynamics to configurations consisting of an infinite number of particles (in both the positive and negative directions) was presented. (Other recent probabilistic studies that have also considered examples of two-sided infinite configurations are \cite{FG, Ferrari}.) At the heart of \cite{CKST} was the observation that the BBS dynamics could be expressed in terms of a certain transformation of an associated path encoding, which we now describe. Firstly, let $\eta\in\{0,1\}^{\mathbb{Z}}$ be a particle configuration (now no longer restricted to a finite number of particles), and then define a path encoding of $\eta$, denoted  $S:\mathbb{Z}\rightarrow\mathbb{Z}$, by setting $S_0:=0$, and
\begin{equation}\label{bbssdef}
S_n-S_{n-1}:=1-2\eta_n,\qquad \forall n\in\mathbb{Z}.
\end{equation}
We then introduce the transformation of reflection in the past maximum $TS:\mathbb{Z}\rightarrow\mathbb{Z}$ via the relation
\begin{equation}\label{bbspitman}
(TS)_n:=2M_n-S_n-2M_0,\qquad \forall n\in\mathbb{Z},
\end{equation}
where $M_n:=\sup_{m\leq n}S_m$ is the past maximum of $S$; this operation of reflection in the past maximum is well-known in the probabilistic literature as Pitman's transformation (after \cite{Pitman}). Clearly, for $TS$ to be well-defined, we require $M_0<\infty$. If this is the case, then we let $T\eta\in\{0,1\}^\mathbb{Z}$ be the configuration given by
\[(T\eta)_n:=\mathbf{1}_{\{(TS)_n-(TS)_{n-1}=-1\}},\qquad \forall n\in\mathbb{Z},\]
(so that $TS$ is the path encoding of $T\eta$). It is possible to check that the map $\eta\mapsto T\eta$ coincides with the original definition of the BBS dynamics $\eta\mapsto\mathcal{T}\eta$ in the finite particle case \cite[Lemma 2.3]{CKST}, and moreover is consistent with an extension to the case of a bi-infinite particle configuration satisfying $M_0<\infty$ from a natural limiting procedure \cite[Lemma 2.4]{CKST}. Thus this description of the BBS could be seen as a natural generalisation of \eqref{udkdv}.

With a view to studying invariant measures for the ultra-discrete Toda lattice and its periodic variant, which is done in the sister article \cite{CSMontreal}, the principal aim of this work is to show that Pitman's transformation is also a suitable means by which to extend the definition of this model to infinite configurations. The one distinction that should be highlighted, however, is that for the ultra-discrete Toda lattice, we incorporate a spatial shift into the transformation of the associated path encoding. We note that the reason for this shift being needed relates to the fact that the ultra-discrete Toda lattice picture does not retain the information about the spatial location of particles that the BBS system does.

To explain the relevant transformation in detail, we first generalise the notion of a configuration for the ultra-discrete Toda lattice. Specifically, we will now consider states to be elements of the form $((Q_n)_{n={N_1}}^{N_2}, (E_n)_{n=N_1}^{N_2-1})\in(0,\infty)^{2(N_2-N_1)+1}$, where $N_1,N_2\in\mathbb{Z}\cup\{-\infty\}\cup\{\infty\}$ satisfy $N_1\leq 1\leq N_2$. By convention, in the case $N_1$ is finite, we set $E_{N_1-1}=\infty$, and in the case $N_2$ is finite, we set $E_{N_2}=\infty$. Furthermore, we suppose that
\begin{equation}\label{noaccumulation}
\sum_{n=N_1}^{1}\left(Q_n+E_{n-1}\right)=\infty,\qquad \sum_{n=1}^{N_2}\left(Q_n+E_{n}\right)=\infty.
\end{equation}
The collection of such elements $((Q_n)_{n={N_1}}^{N_2}, (E_n)_{n=N_1}^{N_2-1})$ will be denoted by $\mathcal{C}$. Given a member of $\mathcal{C}$, the corresponding path encoding is then obtained by concatenating linear segments of length $Q_{N_1},E_{N_1},Q_{N_1+1},E_{N_1+1},\dots,E_{N_2-1},Q_{N_2}$ whose gradient alternates between $-1$ and $1$, with the path being fixed by supposing that the segment arising from $Q_1$ starts at $0$ and has gradient $-1$. More precisely, for $n\in\{N_1,\dots,N_2\}$, let
\begin{equation}\label{indef}
O_n:=\left[\sum_{k=1}^{n-1}\left(Q_k+E_k\right),\sum_{k=1}^{n-1}\left(Q_k+E_k\right)+Q_n\right],
\end{equation}
(where we apply the convention that $\sum_{j=1}^{0}$ is zero, and $\sum_{j=1}^{-n}=-\sum_{j=-n+1}^{0}$ for $n\in\mathbb{N}$,) and set $O=\cup_{n=N_1}^{N_2}O_n$. We then define $S\in C(\mathbb{R},\mathbb{R})$ by setting, analogously to \eqref{bbssdef},
\begin{equation}\label{udsdef}
S_x:=\int_0^x\left(1-2\mathbf{1}_{\{y\in O\}}\right)dy,\qquad \forall x\in\mathbb{R};
\end{equation}
the set of path encodings will be denoted $\mathcal{S}$. (Note that the conditions set out at \eqref{noaccumulation} are not necessary to ensure the function $S$ is well-defined on the whole of $\mathbb{R}$, but do mean that there is no accumulation of local maxima/minima.) It is easy to see that the operation of mapping an element of $\mathcal{C}$ to its path encoding is a bijection. As at \eqref{bbspitman}, we can define $TS:\mathbb{R}\rightarrow \mathbb{R}$ by setting
\begin{equation}\label{pitman}
(TS)_x:=2M_x-S_x-2M_0,\qquad \forall x\in\mathbb{R},
\end{equation}
where $M_x:=\sup_{y\leq x}S_x$ is the past maximum of $S$, and again we require $M_0<\infty$ for this operation to be well-defined. We will denote by $\mathcal{C}^\mathcal{T}$ the set of configurations in $\mathcal{C}$ whose path encodings satisfy the latter condition, and $\mathcal{S}^\mathcal{T}$ the associated set of path encodings. Now, although $TS$ is again an piecewise linear element of $C(\mathbb{R},\mathbb{R})$ with gradient either $-1$ or $1$ (apart from at local maxima or minima), it is not necessarily an element of $\mathcal{S}$. To obtain a path that is in $\mathcal{S}$, we introduce a shift operator. Indeed, for $x\in\mathbb{R}$ and $S\in C(\mathbb{R},\mathbb{R})$, we characterise $\theta^xS$ by
\[(\theta^xS)_y=S_{x+y}-S_x,\qquad\forall y\in\mathbb{R},\]
let $\tau(S):=\inf\{x\geq 0:\:x\in \mathrm{LM}(S)\}$, where $\mathrm{LM}(S)$ is the set of local maxima of $S$ (for the elements of $C(\mathbb{R},\mathbb{R})$ that will be considered in this article, $\tau(S)$ is always well-defined and finite), and set
\[\theta^\tau(S):=\theta^{\tau(S)}(S).\]
We then define an operator $\mathcal{T}$ on $\mathcal{S}^\mathcal{T}$ by the composition of $T$ and $\theta^\tau$, that is
\begin{equation}\label{curlytpath}
\mathcal{T}S:=\theta^\tau(TS),\qquad \forall S\in\mathcal{S}^\mathcal{T}.
\end{equation}
It is possible to check from the definitions that for $S\in \mathcal{S}^\mathcal{T}$, $\tau(TS)=Q_1<\infty$, which means $\mathcal{T}S$ is well-defined, and moreover $\mathcal{T}S\in\mathcal{S}$ with the values of $N_1$ and $N_2$ preserved (see Lemma \ref{lem:invariant}). Hence we can define an associated configuration $(((\mathcal{T}Q)_n)_{n={N_1}}^{N_2}, ((\mathcal{T}E)_n)_{n=N_1}^{N_2-1})\in\mathcal{C}$ (see Figure \ref{todafig} for an example). The following result demonstrates that this procedure is an extension of the dynamics of the ultra-discrete Toda lattice for finite configurations, as given by \eqref{dynamics}. In Subsection \ref{periodicsec}, we moreover show that the transformation of the path encoding we have described also yields the dynamics of the periodic ultra-discrete Toda lattice, see Theorem \ref{mainthm2} in particular.

\begin{figure}[t]
\begin{flushleft}
\hspace{47pt}\includegraphics[width=0.5\textwidth]{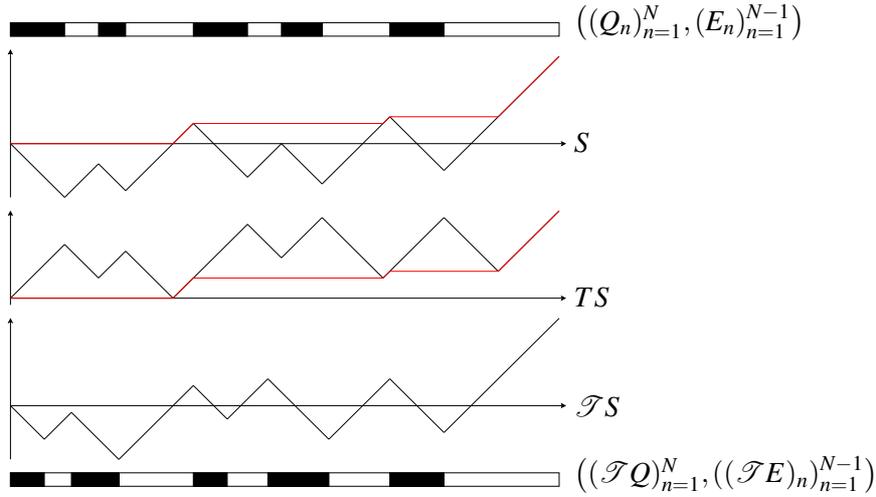}
\rput[tl](0,6.4){$\left((Q_n)_{n=1}^N,(E_n)_{n=1}^{N-1}\right)$}
\rput[tl](0,4.7){$S$}
\rput[tl](0,2.65){$TS$}
\rput[tl](0,1.2){$\mathcal{T}S$}
\rput[tl](0,0.4){$\left((\mathcal{T}Q)_{n=1}^N,((\mathcal{T}E)_n)_{n=1}^{N-1}\right)$}
\end{flushleft}
\caption{Graphical representation of the dynamics of the ultra-discrete Toda lattice in terms of the associated path encodings. NB. The red line in the graphs for $S$ and $TS$ shows the path of $M$.}\label{todafig}
\end{figure}

\begin{thm}\label{mainthm1} Let $N\in\mathbb{N}$, $((Q_n)_{n={1}}^{N}, (E_n)_{n=1}^{N-1})\in(0,\infty)^{2N-1}$, and $S$ represent the associated path encoding, as defined by \eqref{udsdef}. It then holds that $(((\mathcal{T}Q)_n)_{n=1}^N, ((\mathcal{T}E)_n)_{n=1}^{N-1})$, as given by \eqref{dynamics}, has path encoding $\mathcal{T}S$, as defined by \eqref{curlytpath}.
\end{thm}

In the final section of the article, Section \ref{bbsr}, we generalise the above result to a version of the box-ball system on $\mathbb{R}$, as described in \cite{CKST}, where states of the system are described by a certain class of continuous functions, and the dynamics are given by Pitman's transformation \eqref{pitman}. Perhaps most significantly, the result we establish (see Proposition \ref{p31} below) allows us to handle functions $S$ whose gradient is not restricted to $\pm1$, which is potentially relevant in the study of the scaling limits of box-ball systems where the box capacity is not restricted to one, see \cite{CSDual} for background on path encodings for such a model.

\section{Dynamics of the path encoding}

To check the claims of the introduction as we do in this section, it will be convenient to consider a decomposition of $\mathcal{C}^\mathcal{T}$ into subsets. To this end, for $N_1,N_2\in\mathbb{Z}\cup\{-\infty\}\cup\{\infty\}$ with $N_1\leq 1\leq N_2$, let $\mathcal{C}_{N_1,N_2}$ be the subset of $\mathcal{C}$ consisting of elements of the form  $((Q_n)_{n={N_1}}^{N_2}, (E_n)_{n=N_1}^{N_2-1})$, and let $\mathcal{S}_{N_1,N_2}$ be the set of associated path encodings. It is easy to see that the latter set is given by functions $S:\mathbb{R}\rightarrow\mathbb{R}$ of the form
\[S_x=\int_0^x\left(1-2\mathbf{1}_{\{y\in\cup_{n=N_1}^{N_2}[a_n,b_n]\}}\right)dy,\]
where $([a_n,b_n])_{n=N_1}^{N_2}$ is a collection of disjoint intervals such that: $a_1=0$ and $a_n$ increases with $n$; $a_n<b_n$ for each $n$; and only a finite number of intervals intersect any compact set. Although this description is a little cumbersome, it is useful to introduce notation for the points $(a_n)_{n=N_1}^{N_2}$, which correspond to the local maxima of $S$, and $(b_n)_{n=N_1}^{N_2}$, which correspond to the local minima of $S$; we will sometimes write $a_n=a_n(S)$ and $b_n=b_n(S)$ when we wish to stress which path is being considered. We moreover note that if $S$ is the path encoding of $((Q_n)_{n={N_1}}^{N_2}, (E_n)_{n=N_1}^{N_2-1})$, then $[a_n,b_n]$ is the interval $O_n$, as defined at \eqref{indef}. We further define
\[\mathcal{C}^{\mathcal{T}}_{N_1,N_2}:=\mathcal{C}^\mathcal{T}\cap \mathcal{C}_{N_1,N_2},\]
which corresponds to the set of path encodings given by
\[\mathcal{S}^{\mathcal{T}}_{N_1,N_2}:=\mathcal{S}^\mathcal{T}\cap \mathcal{S}_{N_1,N_2}.\]
The key step towards showing that the path transformation $S\mapsto\mathcal{T}S$ yields the dynamics of the ultra-discrete Toda lattice is the following lemma.

\begin{lem}\label{lem:invariant}  Let $N_1,N_2\in\mathbb{Z}\cup\{-\infty\}\cup\{\infty\}$ be such that $N_1\leq 1\leq N_2$.
For any $S \in \mathcal{S}^{\mathcal{T}}_{N_1,N_2}$, it holds that $\mathcal{T}S \in \mathcal{S}_{N_1,N_2}$. Moreover, the elements of $(a_n(\mathcal{T}S),b_n(\mathcal{T}S))_{n=N_1}^{N_2}$ are given by
\begin{align*}
a_n(\mathcal{T}S)&=b_{n}(S)-b_1(S),\\
b_n(\mathcal{T}S)&=\min\left\{b_n(S)+M_{b_n(S)}-S_{b_n(S)},a_{n+1}(S)\right\}-b_1(S),
\end{align*}
where we define $a_{N_2+1}:=\infty$ in the case $N_2<\infty$.
\end{lem}
\begin{proof} In the proof, we will denote $a_n(S),b_n(S)$ by $a_n,b_n$ for simplicity. For $S \in\mathcal{S}^{\mathcal{T}}_{N_1,N_2}$, it holds that
\begin{equation}\label{mbc1}
M_x=M_{a_n} \ \text{for} \ a_n \le x \le b_n
\end{equation}
for each $n$. It is also the case that, for each $n<N_2$, there exists a unique $c_n\in (b_n,a_{n+1}]$ such that
\begin{equation}\label{mbc2}
M_x=\left\{\begin{array}{ll}
               M_{b_n}=M_{a_n}, & \mbox{for }b_n\leq x\leq c_n,\\
               S_x, & \mbox{for }c_n\leq x<a_{n+1}.
             \end{array}\right.
\end{equation}
In the case $N_2<\infty$, the previous claim is also true for $n=N_2$ in the sense that there exists a unique $c_n\in(b_n,\infty)$ such that \eqref{mbc2} holds with $a_{N_2+1}=\infty$. Finally, if $N_1>-\infty$, observe that
\begin{equation}\label{mbc3}
M_x=x+S_{a_{N_1}}-a_{N_1} \ \text{for} \ x \le a_{N_1}.
\end{equation}
From \eqref{mbc1}, \eqref{mbc2} and \eqref{mbc3}, we obtain that
\[TS_x-TS_y=2M_{a_n}-S_x-2M_0-(2M_{a_n}-S_y+2M_0)=S_y-S_x=x-y \ \text{for} \ a_n\le x,y \le b_n,\]
\[TS_x-TS_y=2M_{a_n}-S_x-2M_0-(2M_{a_n}-S_y+2M_0)=S_y-S_x=y-x \ \text{for} \ b_n\le x,y \le c_n,\]
and
\[TS_x-TS_y=2S_x-S_x-2M_0-(2S_y-S_y+2M_0)=S_x-S_y=x-y \ \text{for} \ c_n\le x,y < a_{n+1},\]
for each $n$, where we again interpret $a_{N_2+1}=\infty$ in the case $N_2<\infty$. Moreover, in the case when $N_1>-\infty$, we have that
\begin{align*}
TS_x-TS_y&=2(x+S_{a_{N_1}}-a_{N_1})-S_x-2M_0-(2(y+S_{a_{N_1}}-a_{N_1})-S_y-2M_0)\\
&=2x-S_x-2y+S_y\\
&=x-y\ \text{for} \ x,y \le a_{N_1}.
\end{align*}
In particular, $TS$ is a continuous function of gradient $-1$ or $1$, which is decreasing on the intervals $([b_n,c_n])_{n=N_1}^{N_2}$. It follows from this description that $\tau(TS)=b_1$, and so we obtain $\mathcal{T}S$ is a function in $\mathcal{S}_{N_1,N_2}$ with
\[\left[a_n(\mathcal{T}S),b_n(\mathcal{T}S)\right]=[b_n-b_1,c_n-b_1]\]
for each $n$.

To complete the proof, it remains to check that
\[c_n=\min\left\{b_n+M_{b_n}-S_{b_n},a_{n+1}\right\},\]
with $a_{N_2+1}=\infty$ in the case $N_2<\infty$. For $n<N_2$, we have that $S_{a_{n+1}}=S_{b_n}+a_{n+1}-b_n$, and so
\begin{align*}
\min\left\{b_n+M_{b_n}-S_{b_n},a_{n+1}\right\}&=\min\left\{b_n+M_{b_n}-S_{a_{n+1}}+a_{n+1}-b_n,a_{n+1}\right\}\\
&=a_{n+1}+
\min\{M_{b_n}-S_{a_{n+1}},0\}.
\end{align*}
Now, if $M_{b_n}\geq S_{a_{n+1}}$, then the above expression is equal to $a_{n+1}$. By the definition of $c_n$ at \eqref{mbc2}, it is clear that  $c_n=a_{n+1}$ in this case as well. On the other hand, if $M_{b_n}<S_{a_{n+1}}$, then the above expression is equal to $a_{n+1}+M_{b_n}-S_{a_{n+1}}$. Moreover, we also have that $a_{n+1}-c_{n}=S_{a_{n+1}}-S_{c_n}=S_{a_{n+1}}-M_{b_n}$, i.e.\ $c_n=a_{n+1}+M_{b_n}-S_{a_{n+1}}$, as desired. Finally, when $N_2<\infty$, the argument for $n=N_2$ is similar.
\end{proof}

\subsection{Finite configurations}

We are now ready to prove Theorem \ref{mainthm1}, which we recall concerns the dynamics of the ultra-discrete Toda lattice started from a finite initial configuration.

\begin{proof}[Proof of Theorem \ref{mainthm1}] Let $N\in\mathbb{N}$, $((Q_n)_{n={1}}^{N}, (E_n)_{n=1}^{N-1})\in(0,\infty)^{2N-1}$, and $S$ be the corresponding path encoding, which is clearly an element of $\mathcal{S}^\mathcal{T}_{1,N}$. In particular, this allows us to apply Lemma \ref{lem:invariant} to deduce that $\mathcal{T}S\in \mathcal{S}_{1,N}$, and it remains to prove that
\begin{align*}
(\mathcal{T}Q)_n &= b_{n}(\mathcal{T}S)-a_{n}(\mathcal{T}S), \\
(\mathcal{T}E)_n &= a_{n+1}(\mathcal{T}S)-b_{n}(\mathcal{T}S).
\end{align*}
Denoting $a_n(S),b_n(S)$ by $a_n,b_n$ for simplicity, Lemma \ref{lem:invariant} further implies that the above equations are equivalent to
\begin{align}
(\mathcal{T}Q)_n &= c_n-b_n, \label{eq:Q}  \\
(\mathcal{T}E)_n &=b_{n+1}-c_n, \label{eq:E}
\end{align}
where $c_n:=\min\{b_n+M_{b_n}-S_{b_n},a_{n+1}\}$. Since the dynamics of the ultra-discrete Toda lattice are given by
\begin{align*}
(\mathcal{T}Q)_{n} &=\min \left\{ \sum_{k=1}^n Q_k-\sum_{k=1}^{n-1}(\mathcal{T}Q)_k,E_n\right\}=\min \left\{ b_n+\sum_{k=1}^n Q_k-\sum_{k=1}^{n-1}(\mathcal{T}Q)_k,a_{n+1}\right\}-b_n,\\
(\mathcal{T}E)_n &=Q_{n+1}+E_n-(\mathcal{T}Q)_n=b_{n+1}-b_n-(\mathcal{T}Q)_n,
\end{align*}
we only need to prove that
\begin{align}
\sum_{k=1}^nQ_k -\sum_{k=1}^{n-1}(\mathcal{T}Q)_k = M_{b_n}-S_{b_n} \label{eq:W}
\end{align}
for each $n=1,2,\dots,N$. We will do this by induction. For $n=1$, \eqref{eq:W} holds since
\[\sum_{k=1}^1Q_k -\sum_{k=1}^{0}(\mathcal{T}Q)_k=Q_1=b_1-a_1=b_1=0-(-b_1)=M_{b_1}-S_{b_1}.\]
Next, suppose that \eqref{eq:W} holds for some $n \le N-1$, and note that this implies that we also have \eqref{eq:Q} and \eqref{eq:E} for indices up to this value. In particular, since
\eqref{eq:Q} also holds for $n$,
\begin{align*}
\sum_{k=1}^{n+1}Q_k -\sum_{k=1}^{n}(\mathcal{T}Q)_k = M_{b_n}-S_{b_n} + Q_{n+1}-(\mathcal{T}Q)_{n}=M_{b_n}-S_{b_n}+
b_{n+1}-a_{n+1}-c_n+b_n.
\end{align*}
Hence, to prove \eqref{eq:W} for $n+1$, it is enough to show that
\[M_{b_{n+1}}-S_{b_{n+1}}=M_{b_n}-S_{b_n}+b_{n+1}-a_{n+1}-c_n+b_n.\]
Since $S_{b_{n+1}}-S_{b_n}=(a_{n+1}-b_n)-(b_{n+1}-a_{n+1})$, this is equivalent to showing
\[M_{b_{n+1}}-M_{b_n}=a_{n+1}-c_n.\]
If $S_{a_{n+1}} \le M_{b_n}$, then $M_{b_{n+1}}=M_{b_n}$ and $a_{n+1}=c_n$ holds (cf.\ the proof of Lemma \ref{lem:invariant}). Also, if $S_{a_{n+1}} > M_{b_n}$, then $M_{b_{n+1}}-M_{b_n}=S_{a_{n+1}}-M_{b_n}=S_{a_{n+1}}-S_{c_n}=a_{n+1}-c_n$. So, the desired relation holds.
\end{proof}

To complete this subsection, we show, as a simple corollary of Theorem \ref{mainthm1}, that $\mathcal{T}$ preserves mass.

\begin{cor}\label{masspres} Let $N\in\mathbb{N}$ and $((Q_n)_{n=1}^{N},(E_n)_{n=1}^{N-1})\in (0,\infty)^{2N-1}$. It then holds that
\[\sum_{n=1}^NQ_n=\sum_{n=1}^N(\mathcal{T}Q)_n.\]
\end{cor}
\begin{proof} For $x\geq b_N$, it holds that
\begin{equation}\label{yyy}
S_x=x-2\sum_{n=1}^N(b_n-a_n)=x-2\sum_{n=1}^NQ_n,
\end{equation}
where $S$ is the path encoding of $((Q_n)_{n=1}^{N},(E_n)_{n=1}^{N-1})$. Moreover, for $x\geq c_N$, as defined in the proof of Lemma \ref{lem:invariant}, it holds that $S_x=M_x$. In particular, for $x\geq c_N$, we have $(TS)_x=S_x$. By the definition of $\mathcal{T}S$, it follows that, for $x\geq c_N-b_1$,
\[(\mathcal{T}S)_x=(TS)_{x+b_1}-(TS)_{b_1}=S_{x+b_1}-b_1=x-2\sum_{n=1}^NQ_n.\]
Since we also have, analogously to \eqref{yyy}, that, for $x\geq b_N(\mathcal{T}S)=c_N-b_1$,
\[(\mathcal{T}S)_x=x-2\sum_{n=1}^N(\mathcal{T}Q)_n,\]
the result follows.
\end{proof}

\subsection{Periodic configurations}\label{periodicsec}

In this subsection, we introduce the path encoding of the ultra-discrete periodic Toda lattice. For this model, which was first presented in \cite{KT}, we describe the current state by a vector of the form $((Q_n)_{n=1}^N, (E_n)_{n=1}^N) \in (0,\infty)^{2N}$ for some $N \in \N$. Although it appears we have an extra variable to the non-periodic finite configuration case, this is not so, because we assume
that $L=\sum_{n=1}^N Q_n+ \sum_{n=1}^N E_n$ for some fixed $L >0$. Moreover, in order to define the dynamics,
we suppose that $\sum_{n=1}^N Q_n < \frac{L}{2}$. The collection of such configurations will be denoted by $\mathcal{C}_{(L),N}^{Per}$, and we further associate with $((Q_n)_{n=1}^N, (E_n)_{n=1}^N)\in \mathcal{C}_{(L),N}^{Per}$ an element $((Q_n)_{n \in \Z}, (E_n)_{n \in \Z})\in\mathcal{C}^\mathcal{T}$ by extending the sequences $(Q_n)_{n=1}^N$ and $(E_n)_{n=1}^N$ periodically. Introducing the additional notation $(D_n)_{n=1}^N$ for convenience, the dynamics of the system are given by the following adaptation of
\eqref{dynamics}:
\begin{align}
(\mathcal{T}Q)_n &:= \min \left\{ Q_n-D_n,E_n\right\},  \nonumber \\
(\mathcal{T}E)_n &:= E_n+Q_{n+1}-(\mathcal{T}Q)_n,  \nonumber\\
D_n &:= \min_{0 \le k \le N-1}\sum_{\ell=1}^k (E_{n-\ell}-Q_{n-\ell}).  \label{dynamics:p}
\end{align}
(In this definition, we make use of the periodic extension of $((Q_n)_{n=1}^N, (E_n)_{n=1}^N)$.) Given a state vector  $((Q_n)_{n=1}^N, (E_n)_{n=1}^N)\in \mathcal{C}_{(L),N}^{Per}$, we define an associated path encoding $S$ to be the element of $\mathcal{S}^\mathcal{T}$ associated with its periodic extension $((Q_n)_{n \in \Z}, (E_n)_{n \in \Z})$. We then have the following adaptation of Theorem \ref{mainthm1} to the periodic case. See Figure \ref{periodicfig} for an example realisation of the dynamics, as described by the unshifted Pitman's transformation, and in the original coordinates.

\begin{figure}[t]
\centering
\includegraphics[width = 0.6\textwidth]{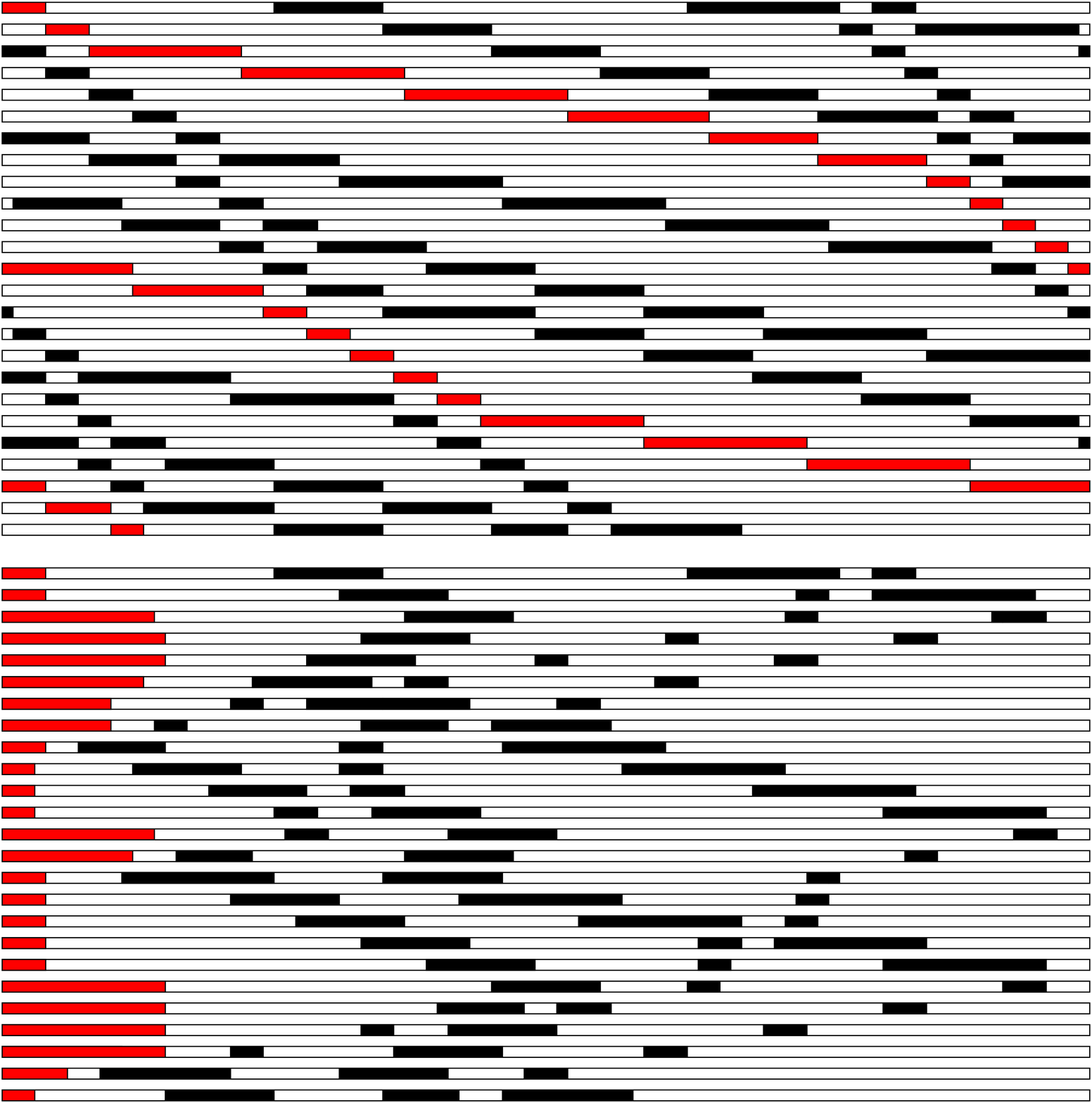}
\caption{First 25 configurations for an example realisation of the periodic ultra-discrete Toda lattice; the top part of the figure shows the spatial position of the configuration, as determined by the unshifted Pitman's transformation, and the bottom part of the figure shows the configuration shifted so that the interval of length $Q_1$ (shown in red in both depictions) is placed first.}\label{periodicfig}
\end{figure}

\begin{thm}\label{mainthm2}
Let $L>0$, $N\in\mathbb{N}$, $((Q_n)_{n=1}^N,(E_n)_{n=1}^{N})\in\mathcal{C}_{(L),N}^{Per}$, and $S$ be the associated path encoding, as described above. It then holds that $(((\mathcal{T}Q)_n)_{n=1}^N,((\mathcal{T}E)_n)_{n=1}^{N})$, as given by \eqref{dynamics:p}, has path encoding $\mathcal{T}S$, as defined by \eqref{curlytpath}. Moreover, it holds that $(((\mathcal{T}Q)_n)_{n=1}^N,((\mathcal{T}E)_n)_{n=1}^{N})\in\mathcal{C}_{(L),N}^{Per}$.
\end{thm}
\begin{proof} In the proof, we again denote $a_n(S),b_n(S)$ by $a_n,b_n$ for simplicity. Since Lemma \ref{lem:invariant} yields that $\mathcal{T}S \in \mathcal{S}$, to establish the first claim of the theorem we only need to prove that
\begin{align*}
(\mathcal{T}Q)_n &= b_{n}(\mathcal{T}S)-a_n(\mathcal{T}S), \\
(\mathcal{T}E)_n &= a_{n+1}(\mathcal{T}S)-b_{n}(\mathcal{T}S).
\end{align*}
By the definition of $\mathcal{T}S$ and Lemma \ref{lem:invariant}, this is equivalent to
\begin{align*}
(\mathcal{T}Q)_n &= c_n-b_n, \\
(\mathcal{T}E)_n &=b_{n+1}-c_n,
\end{align*}
where $c_n:=\min\{b_n+M_{b_n}-S_{b_n},a_{n+1}\}$. Since the dynamics of the ultra-discrete periodic Toda lattice is given by
\begin{align*}
(\mathcal{T}Q)_n &= \min\{Q_n-D_n, E_n\}= \min\{2b_n-a_n-D_n, a_{n+1}\}-b_n, \\
(\mathcal{T}E)_n &= E_n+Q_{n+1}-(\mathcal{T}Q)_n=b_{n+1}-b_n- (\mathcal{T}Q)_n, \\
D_n &= \min_{0 \le k \le N-1}\sum_{\ell=1}^k (E_{n-\ell}-Q_{n-\ell}),
\end{align*}
it will be sufficient to prove that
\begin{align}
b_n-a_{n}-D_n= M_{b_n}-S_{b_n} \label{eq:Wp}
\end{align}
for $n \in \Z$. Now, since we assume that $\sum_{n=1}^N Q_n <\frac{L}{2}$, it holds that, for any $x \in \R$,
\begin{equation}\label{sper}
S_{x+L}-S_x= \sum_{n=1}^NE_n - \sum_{n=1}^NQ_n >0.
\end{equation}
Moreover, since $\{a_n:\:n \in \N\}$ is the set of local maxima of $S$, we have
\[M_{b_n}=\max_{m \le n} S_{a_m} =\max_{n-N+1 \le m \le n}S_{a_m}.\]
On the other hand, for any $n,k$,
\begin{align*}
\sum_{\ell=1}^k (E_{n-\ell}-Q_{n-\ell})
&=\sum_{\ell=1}^k \{(S_{a_{n-\ell+1}}-S_{b_{n-\ell}})+(S_{b_{n-\ell}}-S_{a_{n-\ell}})\}\\
&=\sum_{\ell=1}^{k} (S_{a_{n-\ell+1}}-S_{a_{n-\ell}})\\
&=S_{a_{n}}-S_{a_{n-k}}.
\end{align*}
Hence we have
\begin{align*}
b_n-a_{n}-D_n & =-S_{b_n}+S_{a_{n}}- \min_{0 \le k \le N-1} (S_{a_{n}}-S_{a_{n-k}}) \\
& = -S_{b_n} + \max_{0 \le k \le N-1} S_{a_{n-k}} = -S_{a_n} + \max_{n-N+1 \le m \le n}S_{a_m},
\end{align*}
which establishes \eqref{eq:Wp}. Finally, to prove the claim that $(((\mathcal{T}Q)_n)_{n=1}^N,((\mathcal{T}E)_n)_{n=1}^{N})\in\mathcal{C}_{(L),N}^{Per}$, we simply observe from Lemma \ref{lem:invariant} that
\begin{equation}\label{intervals}
a_{N+1}(\mathcal{T}S)=b_{N+1}-b_1=L.
\end{equation}
\end{proof}

To conclude this subsection, we demonstrate the periodic version of the mass preservation result of Corollary \ref{masspres}.

\begin{cor} Let $L>0$, $N\in\mathbb{N}$ and $((Q_n)_{n=1}^N,(E_n)_{n=1}^{N})\in\mathcal{C}_{(L),N}^{Per}$. It then holds that
\[\sum_{n=1}^NQ_n=\sum_{n=1}^N(\mathcal{T}Q)_n.\]
\end{cor}
\begin{proof} By \eqref{sper}, there exists an $n_0\in 1,\dots,N$ such that $S_{a_{n_0}}=M_{a_{n_0}}$. Hence
\[L-2\sum_{n=1}^NQ_n=S_{a_{n_0+N}}-S_{a_{n_0}}=(TS)_{a_{n_0+N}}-(TS)_{a_{n_0}}=(\mathcal{T}S)_{a_{n_0}+L-b_1}-(\mathcal{T}S)_{a_{n_0}-b_1},\]
where $b_1=b_1(S)$. Applying the fact that the increments of $\mathcal{T}S$ are $L$-periodic and also \eqref{intervals}, we further have that
\[(\mathcal{T}S)_{a_{n_0}+L-b_1}-(\mathcal{T}S)_{a_{n_0}-b_1}=(\mathcal{T}S)_{b_{N+1}-b_1}-(\mathcal{T}S)_{b_{N+1}-b_1-L}=L-2\sum_{n=1}^N(\mathcal{T}Q)_n,\]
from which the result follows.
\end{proof}

\section{Relation to the box-ball system on $\R$}\label{bbsr}

The BBS on $\R$ was introduced in \cite{CKST} as a means to describing scaling limits of discrete models in a high density regime. Specifically, with states of the system being described by continuous functions $S:\mathbb{R}\rightarrow \mathbb{R}$ with $S_0=0$ and $M_0=\sup_{x\leq 0}S_x<\infty$, the collection of which we will denote $\mathcal{S}_0^T$, its dynamics were defined in \cite{CKST} to be given by Pitman's transformation, as at \eqref{pitman}. In this section, we will show that the link between the dynamics of the ultra-discrete Toda lattice and the BBS can be extended beyond the simple case discussed in the previous section.

We will denote the subset of $C(\mathbb{R},\mathbb{R})$ of interest in this section by $\mathcal{S}_{N_1,N_2}$, where $(N_1,N_2)=(-\infty,\infty), (1,\infty),(-\infty,1)$ or $(1,N)$ for some $N \in \mathbb{N}$. NB. This notation conflicts with that used in the previous section, however the set we now introduce includes the former definition, and so there should not be confusion. Specifically, $\mathcal{S}_{N_1,N_2}$ will be the class of $S \in C(\mathbb{R},\mathbb{R})$ with $S_0=0$ that has local maxima at $(a_n)_{n=N_1}^{N_2}$, local minima at $(b_n)_{n=N_1}^{N_2}$ and is otherwise strictly increasing/decreasing, where $a_n <b_n <a_{n+1}$ and $\lim_{n \to \pm \infty}a_n= \pm \infty$ if $N_1$ or $N_2$ is not finite (cf.\ \eqref{noaccumulation}). In the case $(N_1,N_2)=(-\infty,\infty)$, we choose the indices so that $a_0 < 0 \le a_1$. For $S \in \mathcal{S}_{N_1,N_2}$, we let $Q_n=|S_{b_n}-S_{a_n}|$ for $N_1 \le n \le N_2$, $E_n =|S_{a_{n+1}}-S_{b_n}|$ for $N_1 \le n \le N_2-1$, and $\|S\|$ be the total variation of $S$, namely
\[\|S\|_x=\int_0^x |dS_y|,\]
which by assumption is well-defined and finite for each $x\in\mathbb{R}$. Clearly, a path $S\in\mathcal{S}_{N_1,N_2}$ is in one-to-one correspondence with the set of data
\begin{equation}\label{data}
\left(a_1, ((Q_n)_{n=N_1}^{N_2},(E_n)_{n=N_1}^{N_2-1}),\|S\|\right).
\end{equation}
Moreover, if $S\in \mathcal{S}^T_{N_1,N_2}:=\mathcal{S}_{N_1,N_2}\cap\mathcal{S}_0^T$, then it is not difficult to check that $TS\in\mathcal{S}_{N_1,N_2}$ for the same choice of $(N_1,N_2)$. Hence, to analyze the dynamics of BBS for functions in $\mathcal{S}^T_{N_1,N_2}$, we only need to study the dynamics of the data at \eqref{data} under $T$. In the following proposition, we show that these dynamics are essentially determined by the ultra-discrete Toda lattice equation. Hence, since the ultra-discrete Toda lattice equation is integrable, the corresponding initial value problem of BBS on $\R$ restricted to elements of $\mathcal{S}^T_{N_1,N_2}$ can be solved explicitly.

\begin{prop}\label{p31}
(i) Suppose $(N_1,N_2)=(1,\infty),(-\infty,1)$ or $(1,N)$ for some $N \in \mathbb{N}$. Then,
\[a_1(TS)=b_1(S), \qquad \|TS\|=\|S\|\]
and the dynamics of $ ((Q_n)_{n=N_1}^{N_2},(E_n)_{n=N_1}^{N_2-1})$ is given by the ultra-discrete Toda lattice equation \eqref{dynamics}. Namely, $Q_n(TS)= \mathcal{T}Q_n$ and $ E_n(TS)= \mathcal{T}E_n$.\\
(ii) Suppose $(N_1,N_2)=(-\infty,\infty)$. If $b_0(S) <0$, then
\[a_1(TS)=b_1(S), \qquad \|TS\|=\|S\|\]
and the dynamics of $ ((Q_n)_{n=-\infty}^{\infty},(E_n)_{n=-\infty}^{\infty})$ is given by the ultra-discrete Toda lattice equation \eqref{dynamics}. Namely, $Q_n(TS)= \mathcal{T}Q_n$ and $ E_n(TS)= \mathcal{T}E_n$. If $b_0(S) \ge 0$, then
\[a_1(TS)=b_0(S), \qquad \|TS\|=\|S\|\]
and the dynamics of $ ((Q_n)_{n=-\infty}^{\infty},(E_n)_{n=-\infty}^{\infty})$  is given by the ultra-discrete Toda lattice equation \eqref{dynamics} with a spatial shift. Namely, $Q_n(TS)= \mathcal{T}Q_{n-1}$ and $ E_n(TS)= \mathcal{T}E_{n-1}$.
\end{prop}
\begin{proof} (i) We will prove the result in the case $(N_1,N_2)=(1,N)$ for some $N \in \mathbb{N}$; the remaining cases are dealt with similarly. First, by assumption $S$ is strictly increasing on $(-\infty,a_1(S))$ and $(b_1(S),a_2(S))$ (where we interpret $a_2(S)=\infty$ if $N=1$), and strictly decreasing on $(a_1(S),b_1(S))$. It readily follows from the definition of Pitman's transformation at \eqref{pitman} that $TS$ is strictly increasing on $(-\infty,b_1(S))$, and strictly decreasing on $(b_1(S),b_1(S)+\varepsilon)$ for some $\varepsilon>0$. Hence $a_1(TS)=b_1(S)$, as claimed. Second, note that on the set $\{x:\:S_x=M_x\}$ (which consists of a finite collection of closed intervals), we have $dTS_x=dS_x$. Similarly, we have on $\{x:\:S_x<M_x\}$ (which consists of a finite collection of open intervals) that $dTS_x=-dS_x$. Since the set $\{x:\:S_x>M_x\}$ is empty, it follows that $\|TS\|=\|S\|$. Thus it remains to check the claim involving the ultra-discrete Toda lattice equation. To this end, we observe that
\begin{align}
Q_{n}(TS)+E_{n}(TS)&= TS_{a_n(TS)}+TS_{a_{n+1}(TS)}-2TS_{b_n(TS)}\nonumber\\
&=TS_{b_n(S)}+TS_{b_{n+1}(S)}-2\max\left\{M_{a_n(S)},2M_{a_n(S)} - S_{a_{n+1}(S)}\right\}\nonumber\\
&=2M_{a_n(S)}-S_{b_n(S)}+2M_{a_{n+1}(S)}-S_{b_{n+1}(S)}-2\max\left\{M_{a_n(S)},2M_{a_n(S)} - S_{a_{n+1}(S)}\right\}\nonumber\\
&=Q_{n+1}(S)+E_n(S)+2\min\left\{M_{a_{n+1}(S)}-S_{a_{n+1}(S)},M_{a_{n+1}(S)}-M_{a_n(S)}\right\}\nonumber\\
&= Q_{n+1}(S)+E_n(S)\label{bbsr_ud1}
\end{align}
for $n=1,2,\ldots,N-1$, where we have used that ${a_n(TS)}={b_n(S)}$ (which is checked similarly to the case when $n=1$, as described above), $TS_{b_n(TS)}=  \max\{M_{a_n(S)},2M_{a_n(S)} - S_{a_{n+1}(S)}\}-2M_0$ (which follows from \eqref{pitman}), and applied that $-S_{b_n(S)}=E_n(S)-S_{a_{n+1}(S)}$ and $-S_{b_{n+1}(S)}=Q_{n+1}-S_{a_{n+1}(S)}$. Moreover,
\begin{align}
Q_n(TS)&=TS_{a_n(TS)}-TS_{b_n(TS)}\nonumber\\
&= TS_{b_n(S)}  - \max\left\{M_{a_n(S)},2M_{a_n(S)} - S_{a_{n+1}(S)}\right\}\nonumber\\
&=2M_{a_n(S)}-S_{b_n(S)} -\max\left\{M_{a_n(S)},2M_{a_n(S)} - S_{a_{n+1}(S)}\right\}\nonumber\\
&=\min\left\{M_{a_n(S)}-S_{b_n(S)},E_n(S)\right\},\label{fff}
\end{align}
for $n=1,2,\ldots,N$, where we interpret $S_{a_{N+1}(S)}=\infty$. We next claim that
\begin{equation}\label{ilm}
M_{a_n(S)}-S_{b_n(S)}=\sum_{k=1}^nQ_k(S)-\sum_{k=1}^{n-1}Q_k(TS),
\end{equation}
which, similarly to the proof of Theorem \ref{mainthm1}, we prove by induction. The result is obvious by definition for $n=1$. Next, supposing $1\leq n\leq N-1$, we need to show that
\begin{equation}\label{thy}
M_{a_{n+1}(S)}-S_{b_{n+1}(S)}=M_{a_n(S)}-S_{b_n(S)}+Q_{n+1}(S)-Q_{n}(TS).
\end{equation}
The right-hand side here is given by
\begin{align*}
\lefteqn{M_{a_n(S)}-S_{b_n(S)}+S_{a_{n+1}(S)}-S_{b_{n+1}(S)}-\min\left\{M_{a_{n}(S)}-S_{b_{n}(S)},S_{a_{n+1}(S)}-S_{b_{n}(S)}\right\}}\\
&= \max\left\{S_{a_{n+1}(S)},M_{a_{n}(S)}\right\}-S_{b_{n+1}(S)}\\
&=M_{a_{n+1}(S)}-S_{b_{n+1}(S)},\qquad\qquad \qquad \qquad \qquad \qquad \qquad \qquad \qquad \qquad
\end{align*}
thus establishing \eqref{thy}, and consequently \eqref{ilm}. Returning to \eqref{fff}, this yields in turn that
\[Q_n(TS)=\min\left\{\sum_{k=1}^nQ_k(S)-\sum_{k=1}^{n-1}Q_k(TS),E_n(S)\right\}.\]
Together with \eqref{bbsr_ud1}, this establishes that the time evolution of the system is given by ultra-discrete Toda lattice equation \eqref{dynamics}, as desired.\\
(ii) This is proved similarly to (i), with the only difference being that care is needed about the indices depending on whether $b_0(S)$ is $<0$ or $\geq 0$.
\end{proof}

\section*{Acknowledgements}

DC would like to acknowledge the support of his JSPS Grant-in-Aid for Research Activity Start-up, 18H05832, MS would like to acknowledge the support of her JSPS Grant-in-Aid for Scientific Research (B), 16KT0021,
and ST would like to acknowledge the support of his JSPS Grant-in-Aid for Challenging Exploratory Research, 16KT0021.

\providecommand{\bysame}{\leavevmode\hbox to3em{\hrulefill}\thinspace}
\providecommand{\MR}{\relax\ifhmode\unskip\space\fi MR }
\providecommand{\MRhref}[2]{%
  \href{http://www.ams.org/mathscinet-getitem?mr=#1}{#2}
}
\providecommand{\href}[2]{#2}

\end{document}